\documentclass[a4paper,10pt]{scrartcl}
\usepackage[utf8]{inputenc}
\usepackage{amsmath}
\usepackage{amsthm}
\usepackage{amssymb}
\usepackage{hyperref}
\usepackage[ruled,vlined,linesnumbered]{algorithm2e}
\usepackage{caption} 
\usepackage{xcolor}
\usepackage[bottom=4cm,top=2.5cm]{geometry}
\usepackage[backend=biber,maxbibnames=20]{biblatex}
\usepackage{csquotes}

\usepackage{tikz}
\usetikzlibrary{patterns,arrows,positioning,fit,shapes,calc,decorations.markings, decorations, decorations.pathmorphing, backgrounds, positioning}

\usepackage{subcaption}

\DeclareMathOperator{\reals}{\mathbb{R}}
\DeclareMathOperator{\comp}{comp}
\DeclareMathOperator{\est}{est}
\DeclareMathOperator{\umax}{\mathit{u_{max}}}
\DeclareMathOperator{\opt}{OPT}
\DeclareMathOperator{\optlp}{OPT_{LP}}

\theoremstyle{definition}
\newtheorem{defn}{Definition}
\newtheorem{thm}[defn]{Theorem}
\newtheorem{lem}[defn]{Lemma}
\newtheorem{cor}[defn]{Corollary}

\newtheorem{observ}[defn]{Observation}
\newtheorem*{rem}{Remark}
\newtheorem*{claim}{Claim}

\title{A Fast 3-Approximation for the Capacitated Tree Cover Problem with Edge Loads}
\author{Benjamin Rockel-Wolff}
\date{Research Institute for Discrete Mathematics, University of Bonn}

\begin{document}
\maketitle

\begin{abstract}
 The capacitated tree cover problem with edge loads is a variant of the tree cover problem, where we are given facility opening costs, edge costs and loads, as well as vertex loads. We try to find a tree cover of minimum cost such that the total edge and vertex load of each tree does not exceed a given bound. We present an $\mathcal{O}(m\log n)$ time 3-approximation algorithm for this problem.
 
 This is achieved by starting with a certain LP formulation. We give a combinatorial algorithm that solves the LP optimally in time $\mathcal{O}(m\log n)$. Then, we show that a linear time rounding and splitting technique leads to an integral solution that costs at most 3 times as much as the LP solution. Finally, we prove that the integrality gap of the LP is $3$, which shows that we can not improve the rounding step in general.
\end{abstract}

\section{Introduction}

Graph cover problems deal with the following base problem. Given a graph $G$, the task is to find a set of  (connected) subgraphs of $G$, the cover, such that each vertex of $G$ is contained in at least one of the subgraphs. Usually, the subgraphs are restricted to some class of graphs, like paths, cycles or trees. Different restrictions can be imposed on the subgraphs, like a maximum number of edges, or a total weight of the nodes for some given node weights. Recently, Schwartz \cite{schwartz22} published an overview of the literature on different covering and partitioning problems.

We consider the capacitated tree cover problem with edge loads. It is a variation of the tree cover problem that has not been studied so far to the best of our knowledge.

In the capacitated tree cover problem with edge loads, we are given a complete graph $G=(V,E)$, metric edge costs $c:E\rightarrow\reals_+$, vertex loads $b:V\rightarrow [0,1)$, metric edge loads $u:E\rightarrow\reals_{\geq0}$ with $u(e)< u(f)\Rightarrow c(e)\leq c(f)$, and a facility opening cost $\gamma\geq 0$. The task is to find a number of components $k\in\mathbb{N}_{\geq1}$ and a forest $F$ in $G$ consisting of $k$ trees minimizing
\[\sum\limits_{e\in E(F)}c(e)+\gamma k,\] 
such that each tree $T_i$ has total load
\[u(T_i):=\sum_{e\in E(T_i)}u(e)+\sum_{v\in V(T_i)}b(v)\leq1.\]

The capacitated tree cover problem with edge loads is closely related to the facility location problem with service capacities discussed by Maßberg and Vygen in \cite{massberg05}. Their problem uses Steiner trees to connect the nodes, not spanning trees. Furthermore, in their case edge cost and edge load are the same. They make use of this fact to prove a lower bound on the value of an optimum solution. Both problems have important practical applications in chip design. In \cite{held2011} they are called the sink clustering problem and used for clock tree construction. In \cite{bartoschek2014} they are used for repeater tree construction. In these applications terminals and edges have an electrical capacitance. A source can drive only a limited capacitance.

Our problem is also related to other facility location and clustering problems, like the (capacitated) k-center problem (\cite{hochbaum1985}, \cite{khuller2000}) or the k-means problem (\cite{kanungo2002}, \cite{lloyd1982}).

Other tree cover problems include the $k$-min-max tree cover problem and the bounded tree cover problem (\cite{arkin06}, \cite{even04}, \cite{khani14}). In the $k$-min-max tree cover problem, we are given edge weights and want to find $k$ trees such that the maximum of the total weights of the trees is minimized. In the bounded tree cover problem, we are given a bound on the maximum weight of a tree in the cover and try to minimize the number of trees that are required. For these problems Khani and Salavatipour \cite{khani14} gave a 3- and 2.5-approximation respectively. They improve over the previously best known results by Arkin et al.\ \cite{arkin06}, who presented a 4-approximation algorithm for the min-max tree cover problem and a 3-approximation algorithm for the bounded tree cover problem. Even et al.\ \cite{even04} independently gave a 4-approximation algorithm for the min-max tree cover problem.
Furthermore, a rooted version of these problems has been studied. The best known approximation ratio for the capacitated tree case is $7$ and was developed by Yu and Liu \cite{yu19}.

Many algorithms for cycle cover problems are also based on tree cover algorithms (\cite{even04}, \cite{troebst22}, \cite{xu2012}). An example is the capacitated cycle covering problem, where the cover consists of cycles (and singletons) and are given an upper bound on the total nodeweight of the cycles. The task is to minimize the total weight of the cycles plus the facility opening costs. Traub and Tröbst \cite{troebst22} presented a $2+\frac{2}{7}$-approximation for this problem. They use an algorithm for the capacitated tree cover problem as a basis for their $2+\frac{2}{7}$-approximation. In particular, they present a $2$-approximation for the capacitated tree cover problem without edge loads.

\section{Our contribution}
In Section \ref{sec:lp}, we present an LP formulation of the capacitated tree cover problem with edge loads that is based on the formulation in \cite{troebst22}.

Then, we will present a combinatorial algorithm that can optimally solve the LP in time $\mathcal{O}(m \log n)$ in Section \ref{sec:lp_algo}, where $n$ is the number of vertices and $m$ is the number of edges of the graph.

Next, we show how to round the solution to an integral solution in Section \ref{sec:rounding}, employing a splitting technique that runs in linear time from \cite{massberg05}, and show that the resulting integral solution costs at most 3 times as much as the LP-solution. This proves our main theorem:
\begin{thm}
 There is a 3-approximation algorithm for the capacitated tree cover problem with edge loads that runs in time $\mathcal{O}(m\log n)$.
\end{thm}

While the overall approach is similar to the one used in \cite{troebst22}, edges with load require a different algorithm for solving the LP. Furthermore,  we need to be more careful in the analysis of our rounding step.

Finally, in Section \ref{sec:integrality_gap}, we will give an example proving that the integrality gap of our LP is at least 3.

\section{The LP-formulation}
\label{sec:lp}
We may assume that $\gamma\geq c(e)$ for all $e\in E$, as an edge with $c(e)>\gamma$ will never be used in an optimum solution (and could be removed from the solution of the algorithm without increasing the cost).

For simplicity, we will introduce some notation here: For any function $f:A\rightarrow B\subseteq \reals$  from a finite set $A$ into a set $B\subseteq\reals$ and $X\subseteq A$ we write $f(X):=\sum_{x\in X}f(x)$. \\

Given a solution $F$ to our problem with $k$ components $\{T_1,\ldots,T_k\}$, we know that each tree $T_i$ contains exactly $|V(T_i)|-1$ edges and hence $k=|V|-|E(F)|$. Each induced subgraph of $F$ is a forest. So we know \[|E(F[A])|\leq|A|-1\text{ for each }A\subseteq V.\]
Let us now consider the load on the subgraph of $F$, induced by $A\subseteq V$. Each connected component in $F[A]$ can have load at most $1$. So there must be at least $b(A)+u(E(F[A]))$ components in $F[A]$. As each of the components is a tree, the inequality \[|E(F[A])|\leq|A|-(b(A)+u(E(F[A])))\] must be fulfilled.
Using these properties, we can formulate the following LP relaxation of this problem:

\begin{align}
 \min && c^tx+\gamma(|V|-x(E))\\\notag{}\\
 \text{s.t.} &&x(E(G[A]))&\leq|A|-1&\text{ for each }A\subseteq V\\
             &&\sum\limits_{e\in E(G[A])}(1+u(e))x(e)&\leq|A|-b(A)&\text{ for each }A\subseteq V\\
             &&0\leq x(e)&\leq1 &\text{ for each }e\in E
\end{align}
Here $x(e)$ denotes the fractional usage of the edge $e$. We will call an edge $e$ \emph{active} if $x(e)>0$.
The LP can be reformulated by using variables $y(e):=x(e)(1+u(e))$:
\begin{align}
 \min\hspace{0.3cm} && \sum_{e\in E}\frac{c(e)}{1+u(e)}y(e)&+\gamma\left(|V|-\sum_{e\in E}\frac{y(e)}{1+u(e)}\right)\\\notag{}\\
 \text{s.t.}\hspace{0.3cm}&&\sum\limits_{e\in E(A)}\frac{y(e)}{1+u(e)}&\leq|A|-1&\text{ for each }A\subseteq V\label{ineq:forest}\\
 && y(E(G[A]))&\leq|A|-b(A)&\text{ for each }A\subseteq V\label{ineq:load}\\
 &&0\leq y(e)&\leq1+u(e)&\text{ for each }e\in E\label{ineq:bounds}\
\end{align}

For simplicity, we will always consider solutions $x,y$ of both LPs at once. In the following, we will denote by $u_x(e):=x(e)\cdot u(e)$ the fractional load of edge $e$.

\begin{defn}
For a solution $x,y$ to the LP, we define the \emph{support graph} $G_x:=(V,\{e\in E\mid x(e)>0\})$, i.e. the graph consisting of the vertices $V$ and all active edges.

We call an edge \emph{tight} if $y(e)=1+u(e)$, and we call a set $A\subseteq V$ of vertices \emph{tight} if inequality (\ref{ineq:load}) is tight.
\end{defn}
Our goal will be to solve the LP exactly and then round to a forest that may violate the capacity constraints. This increases the edge cost by at most a factor of $2$. In a final step each tree $T$ in the forest with a load $b(V(T))+u(E(T))>1$ can be split into at most $2\cdot\left(b(V(T))+u(E(T))\right)$ trees. This may decrease the edge cost, but loses a factor of $3$ in the number of components, compared to the LP solution.

\section{Solving the LP}
\label{sec:lp_algo}
Altough the LP has an exponential number of inequalities, we can solve it using a simple greedy algorithm, shown in Algorithm \ref{algo:lpsolve}. We will focus on solving the second LP (5) -- (8).

As a first step, we sort the edges $\{e_1,\ldots,e_m\}=E(G)$ such that \[\frac{c(e_1)-\gamma}{1+u(e_1)}\leq\ldots\leq\frac{c(e_m)-\gamma}{1+u(e_m)}.\]
In each iteration, we compute a partition $\mathcal{A}_i\subset2^{V(G)}$ of the vertices of $G$, based on the previous partition $\mathcal{A}_{i-1}$. We initialize $y$ to $0$ and start with $\mathcal{A}_0:=\{\{v\}|v\in V(G)\}$.
Then we iterate through the edges from $e_1$ to $e_m$. For each edge $e_i$, we do the following:

If $e_i$ has endpoints in two different sets of the partition $A_i^1, A_i^2\in\mathcal{A}_{i-1}$, we increase $y(e_i)$ to the maximum possible value. This maximum value is the sum of the slacks of inequalities (\ref{ineq:load}) for the sets $A_i^1$ and $A_i^2$: $|A_i^1|-b(A_i^1)-y(E(G[A_i^1]))+|A_i^2|-b(A_i^2)-y(E(G[A_i^2]))$. However, we assign at most $1+u(e_i)$, such that we do not violate inequality (\ref{ineq:bounds}).
Finally, if we increased $y(e_i)$ by a positive amount, we create the new partition $\mathcal{A}_i$ that arises from $\mathcal{A}_{i-1}$ by removing $A_i^1$ and $A_i^2$ and adding their union.

We set $\mathcal{A}:=\bigcup_{i=1,\ldots,m} \mathcal{A}_i$. Observe that $\mathcal{A}$ is a laminar family. This guarantees that the support graph is always a forest and inequality (\ref{ineq:forest}) is automatically fulfilled.

\begin{algorithm}
\SetKwInOut{Input}{Input}\SetKwInOut{Output}{Output}
\Input{$G,c,u$.}
\Output{$y$ optimum solution of the LP (5) -- (8).}
\BlankLine
 Sort edges such that $\frac{c(e_1)-\gamma}{1+u(e_1)}\leq\ldots\leq \frac{c(e_m)-\gamma}{1+u(e_m)}$\;
 Set $\mathcal{A}_0:=\{\{v\}|v\in V(G)\}$ and $y:=0$\;
 \For{$i=1\ldots m$}{
    \If{there are sets $A_i^1,A_i^2\in \mathcal{A}_{i-1}$ with $e_i\cap A_i^1\neq\emptyset$, $e_i\cap A_i^2\neq\emptyset$ and $A_i^1\neq A_i^2$}{
        $y(e_i):=\min\{1+u(e_i),|A_i^1|-b(A_i^1)-y(E(A_i^1))+|A_i^2|-b(A_i^2)-y(E(A_i^2))\}$\;
        \If{$y(e_i)>0$}{
            $\mathcal{A}_i:=(\mathcal{A}_{i-1}\setminus\{A_i^1,A_i^2\})\cup\{A_i^1\cup A_i^2\}$\;
        }
        \Else{
            $\mathcal{A}_i:=\mathcal{A}_{i-1}$
        }
    }
 }
 \caption{Algorithm for solving the LP (5) -- (8).}
 \label{algo:lpsolve}
\end{algorithm}

\begin{lem}
\label{lem:tight_edges}
Let $x,y$ be the solution computed by Algorithm \ref{algo:lpsolve}. If a set $A\in\mathcal{A}$ from the algorithm is not tight, then all the edges in its induced subgraph $G_{x}[A]$ of the support graph are tight.
\end{lem}
\begin{proof}
Assume this were false. Take a minimal counterexample $A$. As the claim certainly holds for sets consisting only of one vertex ($G_{x}[A]$ has no edges if $|A|=1$), we know that $|A|\geq2$. We can write $A=A^1_i\cup A^2_i$ with their associated edge $e_i$ (for some $i$). We know that $e_i$ has to be tight by line 5, as $A$ is not tight. Otherwise, the algorithm could have increased $y(e_i)$ further. At least one of the subsets $A^1_i$ and $A^2_i$ of $A$ is not tight, otherwise, $A$ were tight. W.l.o.g we may assume that $A^1_i$ is not tight. Then all of its edges are tight, by minimality of $A$. However, then we know that $x(E(G[A^1_i]))=|A^1_i|-1$. Thus,
\begin{multline*}|A^1_i|-b(A^1_i)-y(E(G[A^1_i]))=|A^1_i|-b(A^1_i)-x(E(G[A^1_i]))-u_x(E(G[A^1_i]))\\
=|A^1_i|-b(A^1_i)-(|A^1_i|-1)-u_x(E(G[A^1_i]))=1-(u_x(E(G[A^1_i])+b(A^1_i))\leq1+u(e_i).\end{multline*} This implies that $e_i$ uses up all the slack of $A^1_i$, when it is made tight. Thus, there must be slack remaining on $A^2_i$ and it cannot be tight. As $A$ contains an edge that is not tight in its support graph, this edge must be contained in $A^2_i$. We can conclude that $A^2_i$ is a smaller counterexample. This contradicts the minimality of $A$.
\end{proof} 

\begin{cor}
\label{cor:tight}
 Let $e_i\in E$, $A^1_i$ and $A^2_i$ such that they fulfill the requirements in line 4 of the algorithm. If $e_i$ is tight then exactly one of the following is true:
 \begin{itemize}
  \item Neither $A^1_i$, nor $A^2_i$ are tight.
  \item $A_i^j$ is tight, $A_i^{3-j}$ is not tight and $u(e_i)=u_x(E(G[A_i^{3-j}]))+b(A_i^{3-j})=0$.
 \end{itemize}

\end{cor}
\begin{proof}
If one of the sets is tight. W.l.o.g. we can assume that this set is $A^1_i$. If the other set $A^2_i$ is tight, we know that $y(e_i)=0$, which contradicts our assumption. Otherwise, by Lemma \ref{lem:tight_edges}, all edges in $E(G_{x}[A^2_i])$ are tight. Then $y(e_i)=|A^2_i|-b(A^2_i)-y(E(G[A^2_i]))=1-(u_x(E(G[A^2_i]))+b(A^2_i))\leq1+u(e_i)=y(e_i)$. Hence $u_x(E(G[A^2_i]))+b(A^2_i)= u(e_i)=0$.
\end{proof}

\begin{thm}
\label{thm:lp_algo}
Algorithm \ref{algo:lpsolve} works correctly and has running time $\mathcal{O}(m\log n)$.
\end{thm}
\begin{proof}
The running time is dominated by sorting.\\
For the correctness, we first check that the algorithm outputs a solution to our LP. The minimum in line 5 guarantees that inequality (\ref{ineq:bounds}) is fulfilled. We have already seen that the support graph of our solution is a forest, which means that inequality (\ref{ineq:forest}) is also satisfied.
It remains to check that inequality (\ref{ineq:load}) holds. Each $A\in\mathcal{A}$ fulfills the inequality, when it is introduced by line 5. Since $\mathcal{A}$ is a laminar family, we never change the value of $y(E(G[A]))$ after $A$ has been introduced, so we already know that all $A\in\mathcal{A}$ satisfy inequality (\ref{ineq:load}) when the algorithm is finished.

We define the slack of a set $A\subseteq V$ as the slack of inequality \ref{ineq:load} for that set and denote it by
\[\sigma(A):=|A|-b(A)-y(E(G[A])).\]
Then we can prove that when the algorithm introduces a new set $A$, it has no more slack than each of the joined subsets.
\begin{claim}[1]
\label{claim:subsets_algo}
Let $A_i^1,A_i^2,A\in \mathcal{A}$ such that $A=A_i^1\cup A_i^2$. We claim that $\sigma(A)\leq\sigma(A_i^j)$ for $j=1,2$.
\end{claim}
First, note that 
\begin{align*}
\sigma(A)&=|A|-b(A)-y(E(G[A]))\\
&=|A_i^1|+|A_i^2|-(b(A_i^1)+b(A_i^2))-(y(E(G[A_i^1]))+y(E(G[A_i^2])) + y(e_i))\\
&=\sigma(A_i^1)+\sigma(A_i^2)-y(e_i). \tag{\textborn}\label{eq:A_slack}
\end{align*}
So it is sufficient to show that $y(e_i)\geq\max\{\sigma(A_i^1),\sigma(A_i^2)\}$.
By Line 5 of the algorithm, we have 
\[y(e_i)=\min\{1+u(e_i),\sigma(A_i^1)+\sigma(A_i^2)\}.\]
If we are in the second case, then $y(e_i)=\sigma(A_i^1)+\sigma(A_i^2)$ and we are done.
Otherwise, $y(e_i)=1+u(e_i)$, so $e_i$ is tight. By Corollary \ref{cor:tight}, at least one of the sets $A_i^1$, $A_i^2$ is not tight. W.l.o.g. let this set be $A_i^1$. Then all edges in $E(G_{x}[A^1_i])$ are tight and we have $\sigma(A_i^1)=|A_i^1|-b(A_i^1)-y(E(G[A_i^1]))=1-(b(A_i^1)+u_x(E(G[A_i^1])))\leq1+u(e_i)=y(e_i)$.
For the other set, we have two cases. Either $A_i^2$ is not tight and we can apply the same proof again, or it is tight, then $\sigma(A^2_i)=0\leq y(e_i)$.
This proves the claim.\\

As a next step, we extend Claim (1) to all subsets of $A$.
\begin{claim}[2]
Let $A\in\mathcal{A}$. We claim that each subset $B\subseteq A$ has slack $\sigma(B)\geq\sigma(A)$.
\end{claim}
We prove this by induction on the number of iterations. Let $i'$ be the first index such that the algorithm sets $y(e_{i'})>0$. The claim holds for all iterations before that, because all sets are singletons until this point.
The algorithm creates $A:=A_{i'}^1\cup A_{i'}^2$ with $|A|=2$. Since all subsets of $A$ are in $\mathcal{A}$, we are in the situation of Claim (1) and there is nothing left to prove.

Now let $i>i'$ and assume, that our claim holds in all previous iterations. The algorithm creates $A:=A_i^1\cup A_i^2$. Let $B\subseteq A$. For $|B|=1$, we have $B\in\mathcal{A}$, because all singletons are in $\mathcal{A}$. So we can apply the Claim (1). If $|B|>1$, we split $B$ into two sets $B_1:=B\cap A_i^1$ and $B_2:=B\cap A_i^2$.
We observe that $E(G[B])\setminus (E(G[B_1])\cup E(G[B_2]))\subseteq\{e_i\}$ and thus 
\begin{align*}
\hspace{2.5cm}\sigma(B)&=&&|B|-b(B)-y(E(G[B]))\\
&=&&|B_1|+|B_2|-b(B_1)-b(B_2)-y(E(G[B_1]))\\
&&-&y(E(G[B_2]))-y(E(G[B])\setminus (E(G[B_1])\cup E(G[B_2])))\hspace{2cm}\\
&\geq&&\sigma(B_1)+\sigma(B_2)-y(e_i).
\end{align*}
We use our induction hypothesis on $B_1$ and $B_2$, if they are nonempty, to see that $\sigma(B_1)\geq \sigma(A_i^1)$ and $\sigma(B_2)\geq \sigma(A_i^2)$ respectively. So we can use equation (\ref{eq:A_slack}) to compute
\[\sigma(B)\geq\sum\limits_{j=1,2}\sigma(B_i) - y(e_i)\geq \sum\limits_{j=1,2}\sigma(A_i^j)-y(e_i)=\sigma(A).\]
This proves Claim (2). 

Finally, we observe that for $B_1\subseteq V$ and $B_2\subseteq V$ from different connected components of $G_x$, we have $\sigma(B_1\cup B_2)=\sigma(B_1)+\sigma(B_2)$. This means that we can split any subset of the vertices $B\subseteq V$ into subsets of the toplevel sets of $\mathcal{A}$, which are exactly the connected components of $G_x$. Then we can use Claim (2) on each of the subsets to see that they have nonnegative slack. The observation implies that also $B$ has nonnegative slack. So inequality \ref{ineq:load} is always satisfied. 
$ $\\

Next we want to prove optimality. Assume that $y$ were not optimum. Let $y^*$ be an optimum solution which maximizes the index of the first edge in the order of the algorithm in which $y$ and $y^*$ differ and among those minimizes the difference in this edge. Let this index be denoted by $k$. As the algorithm always sets the values to the maximum that is possible without violating an inequality, we know that $y^*(e_k)<y(e_k)$. 

By the ordering of the algorithm, we know that \[\frac{c(e_k)-\gamma}{1+u(e_k)}\leq \frac{c(e_i)-\gamma}{1+u(e_i)}\] for all $i>k$.
Our goal will be, to find an edge $e_i$ with $i>k$ such that we can increase $y^*(e_k)$ and avoid violating constraints or increasing the objective by decreasing $y^*(e_i)$ in $x^*,y^*$.

Let $G_k$ be the connected component of $e_k$ in the subgraph of $G$ containing only $e_k$ and the active edges with index less than $k$. Define \[\Gamma:=\{e_i\in E(G_{x^*})\mid i>k \text{ and }e_i \text{ incident to } v\in V(G_k)\}.\] Note that $\Gamma\neq\emptyset$, because otherwise, we could increase $y^*(e_k)$ to $y(e_k)$ without violating any constraints. Since $c(e)\leq\gamma$, this would not increase the objective value. 

We will prove that all tight sets containing the vertices of $G_k$ must have a common edge in $\Gamma$.
\begin{claim}[3]
Let $\mathcal{T}:=\{B\subseteq V\mid V(G_k)\subseteq B \text{ and } B \text{ tight}\}$ be the family of tight sets of $x^*,y^*$ containing the vertices of $G_k$. We claim that 
\[\Gamma\cap \bigcap\limits_{B\in\mathcal{T}}E(G_{x^*}[B])\neq\emptyset.\]
\end{claim}

If there is an edge in $\Gamma$ between vertices of $G_k$, then this certainly holds.
Otherwise, we know that $V(G_k)$ is not tight, because the algorithm was able to set $y(e_k)>y^*(e_k)$.\\
\textbf{Observation:} Unions of tight sets are tight and there are no active edges between tight sets.

Let $A,B$ be tight sets. Let $\Delta:=E(G[A\cup B])\setminus (E(G[A])\cup E(G[B]))$ be the edges between $A$ and $B$. We have
\begin{align*}
y(E(G[A\cup B]))&\leq |A\cup B|-b(A\cup B)=|A|-b(A)+|B|-b(B)-(|A\cap B|-b(A\cap B))\\
&=y(E(G[A]))+y(E(G[B]))-(|A\cap B|-b(A\cap B))\\
&=y(E(G[A\cup B]))-y(\Delta)+y(E(G[A\cap B]))-(|A\cap B|-b(A\cap B))\\
&\leq y(E(G[A\cup B]))-y(\Delta),
\end{align*}
where we used LP inequality $(\ref{ineq:load})$ in the first and in the last step.
As $y\geq0$, we must have equality everywhere, and $y(E(G[A\cup B])\setminus (E(G[A])\cup E(G[B])))=0$. This proves our observation.

Let $\mathcal{S}:=\{S_1,\ldots,S_p\}\subseteq \mathcal{T}$ be a set of $p\geq 2$ tight sets containing the vertices of $G_k$ and set $Z:=\bigcup_{S_i\in\mathcal{S}}S_i$. By our observation, $Z$ is tight as well.
We will introduce some notation to write down the proof of the claim. For $a\in\mathbb{N}$, we write $[a]:=\{1,\ldots,a\}$. Then for any index-set $\emptyset\subsetneq I\subseteq[p]$, denote by $V_I:=\bigcap_{i\in I}S_i\setminus V(G_k)$ the vertices of the intersection of the $S_i$ belonging to the indices in $I$ outside of $G_k$, by $E_I:=E(G_{x^*}[V_I])$ the active edges in $V_I$ and by $\Delta_I:=\Gamma\cap E(G_{x^*}[\bigcap_{i\in I}S_i])$ the active edges between $V_I$ and $G_k$. These sets are illustrated in Figure \ref{fig:si_sets}. Furthermore, we denote for $A\subset V(G)$ by $\sigma^*(A):=|A|-b(A)-y^*(E(G[A]))$ the slack of the inequality for $A$ in the optimum solution.

\begin{figure}\centering
\includegraphics{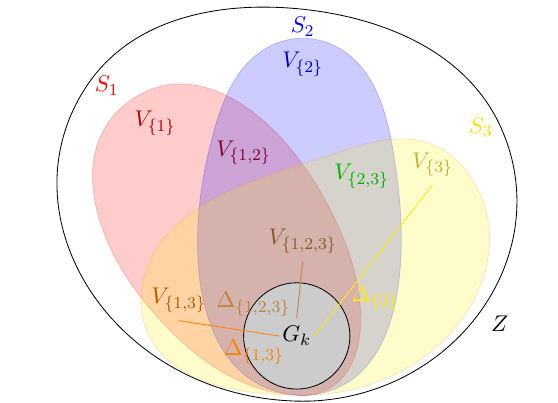}
\caption{An illustration of the sets $Z$, $G_k$, $S_i$, $V_I$ and $\Delta_I$ for $p=3$ and $i\in[3], I\subseteq[3]$. The set $V_{\{1\}}$ for example contains the vertices in the red, the purple, the orange and the brown area, while the set $V_{\{1,2\}}$ contains only the vertices in the purple and the brown area.}
\label{fig:si_sets}
\end{figure}

We will prove a stronger statement than our claim by induction on $p$: For each $p\geq2$ and $I\subseteq[p]$ with $|I|\geq2$, it holds that
\[y^*(\Delta_{I})-\sigma^*(V_{I})=\sigma^*(V(G_k)).\]
Since $V(G_k)$ is not tight, the righthandside must be greater than zero. By our LP-inequalities, $\sigma^*(V_I)\geq0$. Hence $y^*(\Delta_I)>0$ and $\Delta_I\neq\emptyset$.

We start with $p=2$.
Using the tightness of $S_1\cup S_2$, we get 
\[y^*(E(G[Z]))=|Z|-b(Z)=|S_1|-b(S_1)+|S_2|-b(S_2)-(|V_{\{1,2\}}|-b(V_{\{1,2\}}))-(|V(G_k)|-b(V(G_k))).\]
Our observation that edges between tight sets can not be active implies that \[y^*(E(G_k))+y^*(E_{\{1\}})+y^*(E_{\{2\}})-y^*(E_{\{1,2\}})+y^*(\Delta_{\{1\}})+y^*(\Delta_{\{2\}})-y^*(\Delta_{\{1,2\}})=y^*(E(G[Z])).\]
Using the tightness of $S_1$ and $S_2$, we compute
\begin{align*}
y^*(E(G[Z]))&=&&y^*(S_1)+y^*(S_2)-(|V_{\{1,2\}}|-b(V_{\{1,2\}}))-(|V(G_k)|-b(V(G_k)))\\
&=&&2\cdot y^*(E(G_k))+y^*(E_{\{1\}})+y^*(E_{\{2\}})+y^*(\Delta_{\{1\}})+y^*(\Delta_{\{2\}})\\
&&-&(|V_{\{1,2\}}|-b(V_{\{1,2\}}))-(|V(G_k)|-b(V(G_k)))\\
&=&&y^*(E(G[Z])) + y^*(E(G_k)) + y^*(E_{\{1,2\}})+y^*(\Delta_{\{1,2\}})\\
&&-&(|V_{\{1,2\}}|-b(V_{\{1,2\}}))-(|V(G_k)|-b(V(G_k)))\\
&=&&y^*(E(G[Z])) + y^*(\Delta_{\{1,2\}}) - \sigma^*(V_{\{1,2\}}) - \sigma^*(V(G_k)).
\end{align*}
We conclude \[y^*(\Delta_{\{1,2\}}) - \sigma^*(V_{\{1,2\}}) = \sigma^*(V(G_k)).\]

For the induction step, let $p>2$ and assume we have for any index-set $I\subseteq[p]$ with $2\leq|I|<p$ that
\[y^*(\Delta_I)-\sigma^*(V_I)=\sigma^*(V(G_k)).\]
Similarly as in the case $p=2$, we want to write $|Z|-b(Z)$ as a sum of the $|S_i|-b(S_i)$ minus the values that we counted multiple times. We added $|V(G_k)|-b(G_k)$ with each of the $S_i$, but only needed it once. So we have to subtract it $(p-1)$ times. Then for the remainder, we use the following observation:\\
\textbf{Observation:} Let $C_1,\ldots,C_p$ be some sets and $f:\bigcup_{i=1,\ldots,p}C_i\rightarrow\reals$ a function on the elements of the sets. Then 
\[f\left(\bigcup\limits_{i=1,\ldots,p}C_i\right)=\sum\limits_{j=1}^p(-1)^{j-1}\sum_{I\in \binom{[p]}{j}}f\left(\bigcap\limits_{i\in I}C_i\right).\]
We compute
\begin{align*}
y^*(E(G[Z]))&=&&|Z|-b(Z)\\
&=&&\sum\limits_{j=1}^p\left[|S_j|-b(S_j)\right]-(p-1)\cdot(|V(G_k)|-b(G_k))-\sum\limits_{j=2}^p (-1)^{j}\sum\limits_{I\in\binom{[p]}{j}}\left[|V_I|-b(V_I)\right].\\
\end{align*}
Then we use tightness of the $S_j$ for $j=1,\ldots,p$ and the fact that we can split $y^*(E(G[S_j]))=y^*(E_j)+y^*(\Delta_j)+y^*(E(G_k))$ to rewrite this as
\begin{align*}
y^*(E(G[Z]))&=&&\!\!\!\sum\limits_{j=1}^p y^*(E(G[S_j]))-(p-1)\cdot(|V(G_k)|-b(G_k))-\sum\limits_{j=2}^p (-1)^{j}\sum\limits_{I\in\binom{[p]}{j}}\left[|V_I|-b(V_I)\right]\\
&=&&\!\!\!\sum\limits_{j=1}^p \left[y^*(E_j)+y^*(\Delta_j)+y^*(E(G_k))\right]-(p-1)\cdot(|V(G_k)|-b(G_k))\\
 &&\!\!\!-&\!\!\!\sum\limits_{j=2}^p (-1)^{j}\sum\limits_{I\in\binom{[p]}{j}}\left[|V_I|-b(V_I)\right].\\
\end{align*}
We use the reverse argument of the first step to assemble $y^*(Z)$ and add the remaining values:
\begin{align*}
y^*(E(G[Z]))&=&&y^*(Z)+\sum\limits_{j=2}^p (-1)^{j}\sum\limits_{I\in\binom{[p]}{j}}\left[y^*(E_I)+y^*(\Delta_I)\right]+(p-1)y^*(E(G_k))\\
&&-&(p-1)(|V(G_k)|-b(G_k))-\sum\limits_{j=2}^p (-1)^{j}\sum\limits_{I\in\binom{[p]}{j}}\left[|V_I|-b(V_I)\right].\\
\end{align*}
And finally we simplify this to
\begin{align*}
y^*(E(G[Z]))&=&&\!y^*(Z)-(p-1)\sigma^*(V(G_k))
+\sum\limits_{j=2}^p (-1)^{j}\!\!\sum\limits_{I\in\binom{[p]}{j}}\left[y^*(\Delta_I)+y^*(E_I)-(|V_I|-b(V_I))\right]\\
&=&&\!y^*(Z)-(p-1)\sigma^*(V(G_k))+(-1)^p(y^*(\Delta_{[p]})-\sigma^*(V_{[p]}))\\
&&\!\!\!+&\!\sum\limits_{j=2}^{p-1} (-1)^{j}\sum\limits_{I\in\binom{[p]}{j}}\left[y^*(\Delta_I)-\sigma^*(V_I)\right].
\end{align*}
Now we can use our induction hypothesis to retrieve
\[\sum\limits_{j=2}^{p-1} (-1)^{j}\sum\limits_{I\in\binom{[p]}{j}}\left[y^*(\Delta_I)-\sigma^*(V_I)\right]
=\sum\limits_{j=2}^{p-1} (-1)^{j}\sum\limits_{I\in\binom{[p]}{j}}\sigma^*(V(G_k))
=(p-1+(-1)^{p+1})\sigma^*(V(G_k)),\]
which can be put back into our previous equation:
\begin{align*}
\hspace{1cm}y^*(E(G[Z]))&=&&y^*(E(G[Z]))-(p-1)\sigma^*(V(G_k))+(-1)^p(y^*(\Delta_{[p]})-\sigma^*(V_{[p]}))\hspace{1cm}\\
&&+&(p-1+(-1)^{p+1})\sigma^*(V(G_k))\\
&=&&y^*(E(G[Z]))+(-1)^p(y^*(\Delta_{[p]})-\sigma^*(V_{[p]})-\sigma^*(V(G_k)))
\end{align*}
As before, we conclude
\[y^*(\Delta_{[p]})-\sigma^*(V_{[p]})=\sigma^*(V(G_k)).\]

Finally, we can pick an edge $f\in\Gamma$ that is contained in all tight sets that contain the vertices of $G_k$. If $u(f)\leq u(e_k)$, we know that $\frac{1}{1+u(f)}\geq\frac{1}{1+u(e_k)}$. So we can decrease $y^*(f)$ and increase $y^*(e_k)$ by the same amount without violating any constraints. By the ordering of our algorithm, this can not increase the objective value. This would contradict our choice of $y^*$. Hence $u(f)>u(e_k)$. But then $c(f)\geq c(e_k)$ and we could decrease $x^*(f)$ and increase $x^*(e_k)$ without increasing the objective value. Furthermore, we also do not create a violation this way, because $\epsilon\cdot(1+u(f))>\epsilon\cdot(1+u(e_k))$ for $\epsilon>0$. This contradicts our choice of $y^*$ and concludes the proof.
\end{proof}
The support graph of the LP solution computed by Algorithm \ref{algo:lpsolve} is always a forest. Thus, Theorem \ref{thm:lp_algo} implies the following:
\begin{cor}
 There is always a solution $x,y$ to both LPs, such that the support graph $G_x$ is a forest.
\end{cor}

\section{The Rounding Strategy}
\label{sec:rounding}
Now we want to round the LP solution, computed by Algorithm \ref{algo:lpsolve}, to get an integral solution. We do so by rounding up edges $e$ with $x(e)\geq\alpha$, for some $0\leq\alpha\leq1$ to be determined later. All other edges are rounded down. The forest arising from this rounding step may contain components $T$ with $b(V(T))+u(E(T))>1$. These large components will be split into at most $2\cdot(b(V(T))+u(E(T)))$ legal components. We achieve this by using a splitting technique that is often used for these cases, for example in \cite{massberg05} and also in \cite{troebst22}.

We will start by explaining how the splitting technique works in Section \ref{sec:tree_splitting}. However, for the analysis, we only require the result that it is possible to split the trees into $2\cdot(b(V(T))+u(E(T)))$ legal components. 
In Section \ref{subsec:lp_structure}, we will study the LP solution, that we get from Algorithm \ref{algo:lpsolve}. We will exploit the structure of this solution in our analysis. Then we will bound the number of components that we get after rounding and splitting in Section \ref{sec:rounding}. We do this by providing an upper bound on the value of each edge after rounding and splitting. Finally, we determine two different $\alpha$ and the implied bounds on the value of our solution compared to the LP solution. The first bound in Section \ref{sec:first_approx} will depend on the edge loads that are occuring in the instance and is better, when only light edges (with load $<\frac{1}{2}$) occur. The second one in Section \ref{sec:general_approx} will give a factor of $3$ independent of the edge loads.

\subsection{Splitting large trees}
\label{sec:tree_splitting}
Given a rounded component $T$ with total load $b(V(T))+u(E(T))>1$, we want to split this tree into a forest consisting only of trees with total load less or equal to $1$.
\begin{lem}[Maßberg and Vygen 2005\cite{massberg05}]
 There is a linear time algorithm that splits a tree with total load $b(V(T))+u(E(T))>1$ into at most $2\cdot(b(V(T))+u(E(T)))$ legal trees.
\end{lem}
\begin{proof}
We choose an arbitrary vertex $r\in V(T)$ as a root and direct the tree away from $r$. This way, we get an arborescence. During the algorithm, we will maintain a set of trees $\mathcal{T}$, as well as an upper bound on the total load of each tree $l:\mathcal{T}\rightarrow[0,1]$ and an assignment of vertices to trees $t:V(G)\rightarrow\mathcal{T}$. We initialize $\mathcal{T}:=\{\{v\}|v\in V(T)\}$,   $l(\{v\}):=b(v)$ and $t(v):=\{v\}$. Furthermore, we denote for all $v\in V\setminus\{r\}$ by $e_v$ the unique incoming edge of the vertex $v$.

We traverse the arborescence in reverse topological order. At the leaves, we do nothing. At each other vertex $v$, we construct an instance of the bin packing problem. The maximum size of a bin is $1$. Let $C$ be the children of $v$. Then we want to pack $\{v\}\cup C$, where we assign to $v$ the weight $w(v):=l(t(v))$ and to the children $c\in C$ the weight $w(c):=l(t(c))+u(e_c)$.

Then we optimize this bin packing instance. For our purposes, the first fit approach will be enough. For our splitting, we only require, that each bin except for one is packed with at least weight $\frac{1}{2}$ and that the last bin could not have been added to any of the other bins.

Let $B_1,\ldots,B_k$ be the resulting bins ordered in decreasing order of  their weight. For $i=1,\ldots,k$, we do nothing, if $|B_i|=1$. Otherwise, we join the trees $t(p)$ for $p\in B_i$ to a new tree $U$ by shortcutting the paths in $T$ that are connecting them. This can only decrease the load (and cost), because $u$ was metric (and $c$ was metric). Then we remove $t(p)$ from $\mathcal{T}$ for each $p\in B_i$ and we add $U$ to $\mathcal{T}$. We set $l(U):=w(B_i)$. This is an upper bound on the total load of $U$ because we included the weight of the connecting edges in the element weights of the bin packing instance. For the last bin with $i=k$, we also set $t(v):=U$.

Note that due to the reverse topological order, we do not access $t(c)$ or $l(t(c))$ for $c\in C$ again after this point. Also $\mathcal{T}$ will always contain a partition of $T$ and each edge load will have been counted exactly once.

In the end, each tree in $\mathcal{T}$ will correspond to a bin that we produced during our algorithm. Since we always assign the smallest bin to the current vertex, we make sure that it gets passed on to the next iteration. This way, we can guarantee that at most one of the trees has total load less than $\frac{1}{2}$. Furthermore, there was at least one bin, where this tree did not fit. Consequentially the loads of these two trees sum up to more than $1$. This means that on average they have total load at least $\frac{1}{2}$. Since every other tree has total load at least $\frac{1}{2}$, we have created at most $2\cdot(b(V(T))+u(E(T)))$ trees this way.
\end{proof}

\subsection{The general structure of the LP solution}
\label{subsec:lp_structure}
Let $x,y$ be a solution found by the algorithm. Recall that for edges $e\in E(G)$, $u_x(e):=x(e)\cdot u(e)$ was the fractional load of the edge $e$ in our solution. Note that then it holds for each set $A\subseteq V(G)$ and edge $e\in E(G)$ that \[y(E(G[A]))=x(E(G[A]))+u_x(E(G[A]))\text{ and }y(e)=x(e)+x(e)u(e).\]
Without loss of generality, we can assume that $0<x(e)$ for all edges. We simply remove all edges with $x(e)=0$. Then we contract all inclusionwise maximal sets $A\in\mathcal{A}$ such that all edges in their respective induced support graph are tight and set the load of the new vertex to $b(A)+u_x(E(G[A]))$. This only makes the approximation guarantee worse, because these components will have the same value in the rounded solution as in the LP-solution.
Corollary \ref{cor:tight} implies that all remaining edges $e\in E$ with $x(e)=1$, are edges with load $u(e)=0$.
In the remaining graph the following assertions hold:
\begin{enumerate}
 \item $|\{v\}|-b(\{v\})-y(E(G[\{v\}]))=1-b(v)\leq1$ for all $v\in V$.
 \item All $A\in\mathcal{A}$ containing more than $1$ vertex are tight, by Lemma \ref{lem:tight_edges}.
\end{enumerate}
Now, we will take a closer look at the sets $A_i^j$ for $i=1,\ldots,m$ and $j=1,2$. In the following analysis, we will assume without loss of generality that $|A_i^1|\leq |A_i^2|$ and $\sigma(A_i^1)\geq\sigma(A_i^2)$.
By the above assertions, we have for an edge $e_i$ and the two associated sets $A_i^1,A_i^2$, either
\begin{enumerate}
 \item[(i)] both $A_i^1$ and $A_i^2$ contain only one vertex and one of them ($A_i^1$) is not tight, or
 \item[(ii)] $A_i^1$ contains only one vertex and is not tight and $A_i^2$ contains more vertices and is tight
\end{enumerate}
To make the following easier to read, we add the following definitions
\begin{defn}
 Edges that fulfill condition (i) are called \emph{seed edges} and edges that fulfill condition (ii) are called \emph{extension edges}.
For each edge $e_i$ we denote by $v_{e_i}$ the unique vertex in set $A_i^1$.
\end{defn}
Note that every edge $e_i$ is either a seed edge or an extension edge, but this only holds after contracting the sets of tight edges as described above.

Thus, whenever $e_i$ is a seed edge, the algorithm sets \[y(e_i):=|A_i^1|-b(A_i^1)-y(E(G[A_i^1]))+|A_i^2|-b(A_i^2)-y(E(G[A_i^2]))=1-b(A_i^1)+1-b(A_i^2),\] where we use the fact that $E(G[A_i^j])=\emptyset$ for $j=1,2$. Since both $A_i^1$ and $A_i^2$ were singletons, we can conclude \[x(e_i)+u(e_i)x(e_i)=y(e_i)=2-b(A_i^1\cup A_i^2).\] Similarly, for extension edges, we get \[x(e_i)+u(e_i)x(e_i)=1-b(A_i^1).\]
In the analysis of the rounding step, we need some further observations:
\begin{observ}
 Let $T$ be a connected component in $G_x$. Then
 \begin{itemize}
  \item $T$ is a tree.
  \item If $|V(T)|>1$, then $T$ contains exactly one seed edge and all other edges are extension edges.
  \item If $T$ contains a seed edge $e_i$, then $i=\min\limits_{e_j\in E(T)} j$ or in other words, $e_i$ was the first edge of $T$ considered in the algorithm.
 \end{itemize}
\end{observ}
\begin{proof}
 As $G_x$ is a forest, each connected component must be a tree.\\
 Suppose there were two seed edges $e,f\in E(T)$ ($e\neq f$). Let $A$ be the set that the algorithm creates, when increasing $e$, and $B$ the same for $f$. By the definition of seed edges, $|A|=|B|=2$. Let $C=A_i^1\dot{\cup} A_i^2$ be the smallest set in the laminar family $\mathcal{A}$ with $A,B\subset C$. Then $A\subseteq A_i^k$ and $B\subseteq A_i^{3-k}$ for some $k\in \{1,2\}$.
 Hence $|A_i^j|\geq 2$ for $j=1,2$ and both must be tight. Thus $y(e_i)=0$, contradicting the fact that Algorithm \ref{algo:lpsolve} has joined $A_i^1$ and $A_i^2$. This implies that $T$ can only contain one seed edge.\\
 Let $|V(T)|>1$ and $e_i$ be the first edge in $T$ that was considered during the algorithm. Then $|A_i^j|=1$ for $j=1,2$. So by definition, it is a seed edge.
\end{proof}

\subsection{Analyzing the rounding step}
\label{sec:rounding_analysis}
 First note that by our rounding procedure, the sum of the edge-weights can increase by at most $\frac{1}{\alpha}$. So for the edge-weights it is sufficient to make sure that $\alpha\geq\frac{1}{2}$ and the main difficulty is to bound the number of components.
 
 Before we choose $\alpha$, let us estimate how many components we get after rounding and splitting. To do this, we take a look at the connected components after rounding. Let $T$ be such a component. We denote by $\comp(T)$ the number of connected components we need to split $T$ into. 
 
 Let $C^*$ be the set of components before splitting and $C$ be the set of components after splitting. Our goal here is to estimate $|C|$ by a contribution $\est(e)$ of each edge $e\in E(G)$, such that the number of components after splitting is
\begin{align*}
 |C|=\sum\limits_{T\in C^*}\comp(T)\leq\sum\limits_{T\in C^*}|V(T)|-\sum\limits_{e\in E(G)}\est(e)=|V(G)|-\sum\limits_{e\in E(G)}\est(e)
\end{align*}
 There are three cases:
\begin{enumerate}
 \item \emph{singletons}: $T$ consists of only one vertex.
 \item \emph{good trees}: $T$ consists of more than one vertex and $u(E(T))+b(V(T))\leq1$
 \item \emph{large trees}: $T$ consists of more than one vertex and $u(E(T))+b(V(T))>1$
\end{enumerate}$ $\\
\textbf{Case 1}: $T$ is a singleton. Its number of components is \[\comp(T):=1=|V(T)|-0.\]
\textbf{Case 2}: $T$ is a good tree. So we can keep this component for a solution to the problem. The number of components is
\[\comp(T):=1=|V(T)|-(|V(T)|-1)\leq|V(T)|-\sum\limits_{e\in E(T)}\left[1-2b(v_e)-2u(e)\right].\]
For all $e\in E(T)$, we set $\est(e):=1-2b(v_e)-2u(e)$.\\\\
\textbf{Case 3}: $T$ is a large tree. So we have to split this component to get a feasible solution.
Denote by $e'$ the edge in $T$ with the lowest index according to the sorting of the algorithm. Let $\bar{v}\neq v_{e'}$ be incident to $e'$. Note that this does not have to be a seed edge, as the components after rounding do not necessarily contain a seed edge. We rewrite the number of components:
\begin{align*}
\comp(T)&\leq2\cdot(u(E(T))+b(V(T)))\\
&=|V(T)|-\left[2-2b(v_{e'})-2u(e')-2b(\bar{v})\right]-\sum\limits_{e'\neq e\in E(T)}\left[1-2b(v_e)-2u(e)\right].
\end{align*}
If $T$ contains a seed edge, then this edge is $e'$. This means that the number of components can be estimated by edges in $T$. We set $\est(e'):=2-2b(v_{e'})-2u(e')-2b(\bar{v})$ and $\est(e):=1-2b(v_e)-2u(e)$ for all $e\in E(T)\setminus\{e'\}$.

Otherwise $T$ only consists of extension edges. In this case, we write
\begin{multline*}\left[2-2b(v_{e'})-2u(e')-2b(\bar{v})\right]+\sum\limits_{e'\neq e\in E(T)}\left[1-2b(v_e)-2u(e)\right]\\
 =[1-2b(\bar{v})]+\sum\limits_{e\in E(T)}\left[1-2b(v_e)-2u(e)\right].
\end{multline*}
Then, we set $\est(e):=1-2b(v_{e})-2u(e)$ for all $e\in E(T)$. However, in this case we need to account for the additional $1-2b(\bar{v})$. To do so, we call the edge incident to $\bar{v}$ that is not contained in $T$ a \emph{filler edge}. For all filler edges $\{v,w\}=e\in E(G)$, we w.l.o.g. assume that $e$ is a filler edge for the component that contains $v$ and set
\[\est(e):=\begin{cases}2-(b(v)+b(w)),&\text{if } e \text{ is the filler edge of two components}\\1-b(v),&\text{ otherwise.} \end{cases}\]
For all edges not considered before, we set $\est(e):=0$.

Now we have that
\[|C|\leq|V(G)|-\sum\limits_{e\in E(G)}\est(e)\]
$ $\\
Our next goal is to find a lower bound on $\sum\limits_{e\in E(G)}\est(e)$.
\subsubsection{Lower bounds for the extension edges}
We start with the simpler case of extension edges. An overview over the cases in which they can appear is shown in Figure \ref{fig:extension_edge_appearances}. Let $e$ be an extension edge. If it appears inside a good tree or a large tree. Then
\begin{align*}
 \est(e)&=1-2b(v_e)-2u(e)\\
 &=1-2(1-x(e)-x(e)u(e)+u(e))\\
 &=1-2+2x(e)+2x(e)u(e)-2u(e)\\
 &=2x(e)-1-2u(e)(1-x(e)).
\end{align*}
If it is incident to a singleton or a good tree, we can estimate
\[\est(e)=0\geq2x(e)-1-(2x(e)-1).\]
If it is a filler edge, we can estimate
\[\est(e)=1-2b(v_e)=1-2(1-x(e)-x(e)u(e))=2x(e)-1+x(e)u(e)\geq2x(e)-1.\]

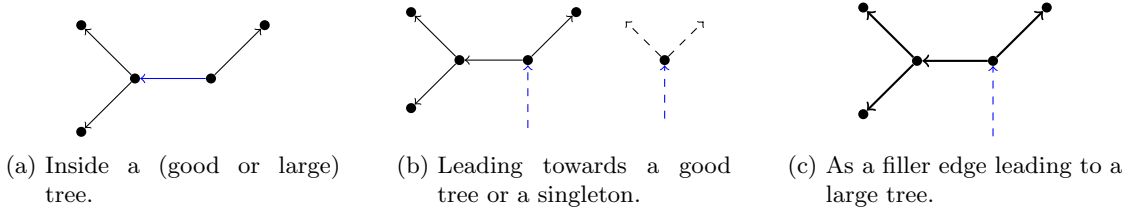
\begin{figure}
\begin{subfigure}[t]{0.3\textwidth}
\centering
   \begin{tikzpicture}
  \node[circle,fill=black,scale=0.4] (a) at (0,0){};
  \node[circle,fill=black,scale=0.4] (b) at (-0.707,0.707){};
  \node[circle,fill=black,scale=0.4] (c) at (-0.707,-0.707){};
  \node[circle,fill=black,scale=0.4] (d) at (1,0){};
  \node[circle,fill=black,scale=0.4] (e) at (1.707,0.707){};
 
  \draw[->] (a) -- (b);
  \draw[->] (a) -- (c);
  \draw[->,color=blue] (d) -- (a);
  \draw[->] (d) -- (e);
 \end{tikzpicture}
 \caption{Inside a (good or large) tree.}
\end{subfigure}
\hfill
\begin{subfigure}[t]{0.3\textwidth}
\centering
 \begin{tikzpicture}[scale=0.9]
  \node[circle,fill=black,scale=0.4] (a) at (0,0){};
  \node[circle,fill=black,scale=0.4] (b) at (-0.707,0.707){};
  \node[circle,fill=black,scale=0.4] (c) at (-0.707,-0.707){};
  \node[circle,fill=black,scale=0.4] (d) at (1,0){};
  \node[circle,fill=black,scale=0.4] (e) at (1.707,0.707){};
 
  \draw[->] (a) -- (b);
  \draw[->] (a) -- (c);
  \draw[->] (d) -- (a);
  \draw[->] (d) -- (e);
  \draw[dashed,->,color=blue] (1,-1) -- (d);
  
  \node[circle,fill=black,scale=0.4] (k) at (3,0){};
  \node (l) at (2.293,0.707){};
  \node (m) at (3.707,0.707){};
  \node (n) at (3,-1){};
  
  \draw[dashed,->] (k) -- (l);
  \draw[dashed,->] (k) -- (m);
  \draw[dashed,->, color=blue] (n) -- (k);
 \end{tikzpicture}
 \caption{Leading towards a good tree or a singleton.}
\end{subfigure}
\hfill
\begin{subfigure}[t]{0.3\textwidth}
\centering
 \begin{tikzpicture}
  \node[circle,fill=black,scale=0.4] (a) at (0,0){};
  \node[circle,fill=black,scale=0.4] (b) at (-0.707,0.707){};
  \node[circle,fill=black,scale=0.4] (c) at (-0.707,-0.707){};
  \node[circle,fill=black,scale=0.4] (d) at (1,0){};
  \node[circle,fill=black,scale=0.4] (e) at (1.707,0.707){};
 
  \draw[thick,->] (a) -- (b);
  \draw[thick,->] (a) -- (c);
  \draw[thick,->] (d) -- (a);
  \draw[thick,->] (d) -- (e);
  \draw[dashed,->,color=blue] (1,-1) -- (d);
 \end{tikzpicture}
 \caption{As a filler edge leading to a large tree.}
\end{subfigure}
 \caption{The cases in which extension edges can occur. Dashed edges have been rounded down, while solid ones have been rounded up. Thick edges belong to a large tree. For each edge $e$ the arrowhead points towards $v_e$.}
 \label{fig:extension_edge_appearances}
\end{figure}

\subsubsection{Lower bounds for the seed edges}
Next we consider seed edges. An overview over the cases in which they can appear is shown in Figure \ref{fig:seed_edge_appearances}. Let $e$ be a seed edge. $e=\{v_e,\bar{v}\}$ can not be contained in a singleton. It can only be contained in a good tree, if $e$ is tight, as otherwise we have \[1+u(e)>y(e)=1-b(v_e)+1-b(\bar{v})\Leftrightarrow b(v_e)+b(\bar{v})+u(e)>1.\]
If it is tight, we are in the second case of Corollary \ref{cor:tight} and $u(e)=0$. Recall that then one of the sets $A_i^1$, $A_i^2$ was tight for $e=e_i$. By our labelling this was $A_i^2=\bar{v}$. So, we know that $b(\bar{v})=1$.
We estimate
\begin{align*}
 \est(e)&=1-2b(v_{e})-2u(e)\\
 &=1-2(2-x(e)-x(e)u(e)-b(\bar{v})+u(e))\\
 &=1-4+2x(e)+2b(\bar{v})\\
 &=2x(e)-3+2b(\bar{v})\\
 &\geq2x(e)-2.
\end{align*}
If it is contained a large tree, it was the first edge considered in this component. We estimate 
\begin{align*}
 \est(e)&=2-2b(v_{e})-2u(e)-2b(\bar{v})\\
 &=2-2(2-x(e)-x(e)u(e)+u(e))\\
 &=2-4+2x(e)+2x(e)u(e)-2u(e)\\
 &=2x(e)-2-2u(e)(1-x(e)).
\end{align*}
Otherwise both endpoints are incident to different components. This means that it was rounded down. If these components are singletons or good trees, we can estimate
\[\est(e)=0\geq2x(e)-2.\]
If both are large trees, then e is a filler edge for both and we have
\[\est(e)=2-2(b(v_e)+b(\bar{v}))=2-2(2-x(e)-x(e)u(e))=2x(e)-2+2u(e)x(e)\geq2x(e)-2.\]
The last case is that $e$ is incident to one large tree and a good tree or a singleton. This means it is a filler edge for only one endpoint. W.l.o.g. let this endpoint be $v_e$. We set $y_1:=(1+u(e))$, $x_1:=1-b(v_e)$ and $y_2:=(1+u(e))$, $x_2:=1-b(\bar{v})$. For a later estimate note that then $x_2\leq\alpha$ as $x(e)\leq\alpha$.
We can estimate 
\begin{align*}
\est(e)&=1-2b(v_e)\\
&=1-2(1-x_1-u(e)x_1)\\
&=1-2+2x_1+2u(e)x_1\\
&=2x_1-1+2u(e)x_1\\
&=2x_1+2x_2-2+ 2u(e)x_1+1-2x_2\\
&\geq2x(e)-2-(2x_2-1).
\end{align*}
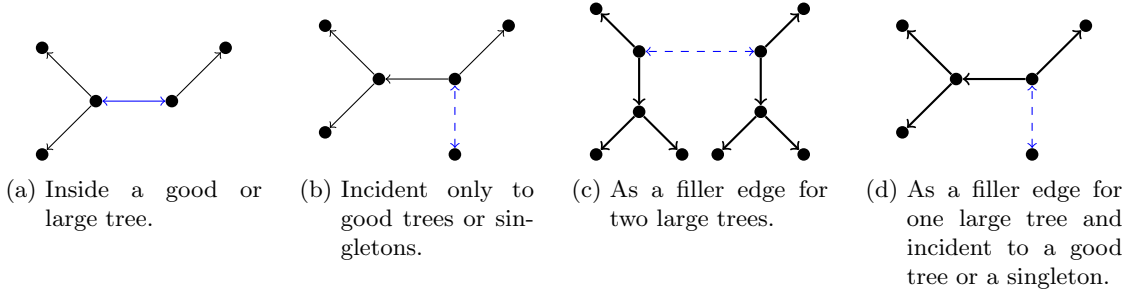
\begin{figure}
\begin{subfigure}[t]{0.23\textwidth}
\centering
  \begin{tikzpicture}
  \node[circle,fill=black,scale=0.5] (a) at (0,0){};
  \node[circle,fill=black,scale=0.5] (b) at (-0.707,0.707){};
  \node[circle,fill=black,scale=0.5] (c) at (-0.707,-0.707){};
  \node[circle,fill=black,scale=0.5] (d) at (1,0){};
  \node[circle,fill=black,scale=0.5] (e) at (1.707,0.707){};
 
  \draw[,->] (a) -- (b);
  \draw[,->] (a) -- (c);
  \draw[,<->,color=blue] (d) -- (a);
  \draw[,->] (d) -- (e);
 \end{tikzpicture}
 \caption{Inside a good or large tree.\vspace{3\baselineskip}}
\end{subfigure}
\hfill
\begin{subfigure}[t]{0.21\textwidth}
\centering
  \begin{tikzpicture}
  \node[circle,fill=black,scale=0.5] (a) at (0,0){};
  \node[circle,fill=black,scale=0.5] (b) at (-0.707,0.707){};
  \node[circle,fill=black,scale=0.5] (c) at (-0.707,-0.707){};
  \node[circle,fill=black,scale=0.5] (d) at (1,0){};
  \node[circle,fill=black,scale=0.5] (e) at (1.707,0.707){};
  \node[circle,fill=black,scale=0.5] (f) at (1,-1){};
 
  \draw[->] (a) -- (b);
  \draw[->] (a) -- (c);
  \draw[->] (d) -- (a);
  \draw[->] (d) -- (e);
  \draw[dashed,<->,color=blue] (f) -- (d);
   \end{tikzpicture}
   \caption{Incident only to good trees or singletons.\vspace{\baselineskip}}
\end{subfigure}
\hfill
\begin{subfigure}[t]{0.23\textwidth}
\centering
  \begin{tikzpicture}[rotate=90,scale=0.8]
    \node[circle,fill=black,scale=0.5] (a) at (0,0){};
  \node[circle,fill=black,scale=0.5] (b) at (-0.707,0.707){};
  \node[circle,fill=black,scale=0.5] (c) at (-0.707,-0.707){};
  \node[circle,fill=black,scale=0.5] (d) at (1,0){};
  \node[circle,fill=black,scale=0.5] (e) at (1.707,0.707){};
    \node[circle,fill=black,scale=0.5] (f) at (0,-2){};
  \node[circle,fill=black,scale=0.5] (g) at (-0.707,-1.293){};
  \node[circle,fill=black,scale=0.5] (h) at (-0.707,-2.707){};
  \node[circle,fill=black,scale=0.5] (i) at (1,-2){};
  \node[circle,fill=black,scale=0.5] (j) at (1.707,-2.707){};
 
  \draw[thick,->] (a) -- (b);
  \draw[thick,->] (a) -- (c);
  \draw[thick,->] (d) -- (a);
  \draw[thick,->] (d) -- (e);
  \draw[dashed,<->,color=blue] (i) -- (d);
  \draw[thick,->] (f) -- (g);
  \draw[thick,->] (f) -- (h);
  \draw[thick,->] (i) -- (f);
  \draw[thick,->] (i) -- (j);
  \end{tikzpicture}
  \caption{As a filler edge for two large trees.\vspace{2\baselineskip}}
\end{subfigure}
\hfill
\begin{subfigure}[t]{0.23\textwidth}
\centering
  \begin{tikzpicture}
  \node[circle,fill=black,scale=0.5] (a) at (0,0){};
  \node[circle,fill=black,scale=0.5] (b) at (-0.707,0.707){};
  \node[circle,fill=black,scale=0.5] (c) at (-0.707,-0.707){};
  \node[circle,fill=black,scale=0.5] (d) at (1,0){};
  \node[circle,fill=black,scale=0.5] (e) at (1.707,0.707){};
  \node[circle,fill=black,scale=0.5] (f) at (1,-1){};
 
  \draw[thick,->] (a) -- (b);
  \draw[thick,->] (a) -- (c);
  \draw[thick,->] (d) -- (a);
  \draw[thick,->] (d) -- (e);
  \draw[dashed,<->,color=blue] (f) -- (d);
   \end{tikzpicture}
   \caption{As a filler edge for one large tree and incident to a good tree or a singleton.}
\end{subfigure}
 \caption{The cases in which seed edges can occur. Dashed edges have been rounded down, while solid ones have been rounded up. Thick edges belong to a large tree. For each extension edge $e$ the arrowhead points towards $v_e$. Seed edges have arrowheads on both ends.}
 \label{fig:seed_edge_appearances}
\end{figure}
Now almost all estimates are of the same form. 

\subsubsection{Summary of the estimates}
\label{sec:estimate_summary}
Before we choose $\alpha$ and derive the approximation guarantee, let us summarize the derived estimates.
\begin{align*}
\est(e)\geq\\
&&\text{\textbf{seed edges}}\\
&2x(e)-2-\textcolor{blue}{0}&\text{incident to good trees or singletons,}\\
&&\text{in a good tree, filler for both ends}\\
&2x(e)-2-\textcolor{blue}{(2x_2-1)}&\text{filler for one end}\\
&2x(e)-2-\textcolor{blue}{2u(e)(1-x(e))}&\text{in a large tree}\\\\
&&\text{\textbf{extension edges}}\\
&2x(e)-1-\textcolor{blue}{0}&\text{filler edge}\\
&2x(e)-1-\textcolor{blue}{(2x(e)-1)}&\text{incident to good tree or singleton}\\
&2x(e)-1-\textcolor{blue}{2u(e)(1-x(e))}&\text{inside a component}
\end{align*}
The base part, which is left in black, now sums up to at most $2x(E(G))-|V(G)|$, because there is exactly one seed edge for every component in $G$ that is not a singleton. So it remains to estimate the parts marked in blue.
Our goal will be to estimate this part in terms of $|V(G)|-x(E(G))$. That is, find a $\beta$, such that we have $\text{\textcolor{blue}{``sum of blue parts''}}\leq\beta(|V(G)|-x(E(G))$.
We will achieve this by first estimating for each $\{v_e,w\}\in E(G_{x^*})$
that \[\text{\textcolor{blue}{``blue part''}}\leq
\begin{cases}
\beta(b(v_e)+b(w)+x(e)u(e))&\text{for seed edges}\\
\beta(b(v_e)+x(e)u(e))&\text{for extension edges}.
\end{cases}
\]
Then, we can use that to sum up the estimates of the differences 
\[\text{\textcolor{blue}{``sum of blue parts''}}\leq \beta(b(V(G))+u(x(E(G))))\leq\beta(|V(G)|-x(E(G))),\]
where the last inequality follows directly from the LP-inequalities.

In total, we are left with
\begin{align*}|C|&\leq |V(G)|-\sum\limits_{e\in E(G)}\est(e)\\
&\leq |V(G)| - (2x(E(G))-|V(G)|-\beta(|V(G)|-x(E(G))))\\
&= (2+\beta)(|V(G)|-x(E(G)))
\end{align*}

\subsubsection{A first approximation guarantee}
\label{sec:first_approx}
Let $\umax:=\max\limits_{e\in E(G)}u(e)$
\begin{lem}
If we set $\alpha:=\frac{1}{2}$, the number of components after splitting is at most \[(2+2\umax)(|V(G)|-x(E(G))).\]
\end{lem}
\begin{proof}
 We will show the following claim:
 
 If $\alpha=\frac{1}{2}$, then \[\est(e)\geq2x(e)-2-2\umax(b(v)+b(w)+x(e)u(e))\] for all seed edges $e=\{v,w\}$ and \[\est(e)\geq2x(e)-1-2\umax(b(v)+u(e)x(e))\] for all extension edges $e=\{v,w\}$.
 
 These estimates are exactly of the form, we required in Section \ref{sec:estimate_summary}. Thus, they directly imply the statement of the lemma.

 First we observe that $-u(e)(1-x(e))\geq-\umax(1-x(e))$. Second, that for seed edges $e=\{v,w\}$, we have \[1-x(e)=b(v)+b(w)+u(e)x(e)-1\leq b(v)+b(w)+u(e)x(e)\] and for extension edges $e=\{v_e,w\}$, \[1-x(e)=b(v_e)+u(e)x(e).\] Furthermore of course $0\geq-2\umax(b(v)+b(w)+u(e)x(e))$ and $0\geq-2\umax(b(v_e)+u(e)x(e))$. With these observations the claim already holds for most edges. It remains to show that it holds for extension edges $e$ that are incident to good trees or singletons. These edges are rounded down, as they are in no component. This means $x(e)\leq\frac{1}{2}$. We have \[\est(e)\geq2x(e)-1-(2x(e)-1)\geq2x(e)-1\geq2x(e)-1-2\umax(b(v_e)+u(e)x(e)).\] And we need to show the inequality for seed edges that are filler edges for only one large tree. These are also rounded down. There we estimated $\est(e)\geq2x(e)-2-(2x_2-1)$. With $x_2\leq\frac{1}{2}$, we get \[\est(e)\geq2x(e)-2\geq2x(e)-2-2\umax(b(v)+b(w)+u(e)x(e)).\] This proves the claim.
\end{proof}

\subsubsection{A general approximation guarantee}
\label{sec:general_approx}
However, these estimates were not optimal, as we may choose $\alpha$ in a better way. Specifically, we will choose $\alpha:=\frac{2}{3}$ to achieve a 3-approximation.

We want to determine $\beta\leq1$ optimal, such that we get a $2+\beta$-approximation. The only step, that required $\alpha\leq\frac{1}{2}$ in the previous proof was to estimate $-(2x(e)-1)\geq0$ for edges that were not rounded up. However, if we want a $2+\beta$-approximation, we may choose $\alpha$, such that $-(2x(e)-1)\geq-\beta(1-x(e))$ (for edges that are rounded down).
\begin{align*}
 && -(2x(e)-1)&\geq-\beta(1-x(e))\\
 \Leftrightarrow&& 1-2x(e)&\geq-\beta+\beta x(e)\\
 \Leftrightarrow&&1+\beta&\geq(2+\beta)x(e)\\
 \Leftrightarrow&&\frac{1+\beta}{2+\beta}&\geq x(e)
\end{align*}
This means for all edges $e$ that we round down, we need \[x(e)\leq\frac{1+\beta}{2+\beta}.\]
So we set $\alpha:=\frac{1+\beta}{2+\beta}$. Now we know for all edges $e$ that we round up, $x(e)\geq\frac{1+\beta}{2+\beta}$. In this case we have $1-x(e)\leq\frac{1}{2+\beta}$. So we can estimate
\[x(e)\geq\frac{1+\beta}{2+\beta}\geq(1+\beta)(1-x(e)).\]
Then we can use this to estimate
\[-2u(e)(1-x(e))\geq-\frac{2}{1+\beta}u(e)x(e),\]
as the only edges, where the term $-2u(e)(1-x(e))$ appeared, were edges that we rounded up.
Of course then for seed edges this implies \[-\frac{2}{1+\beta}u(e)x(e)\geq-\frac{2}{1+\beta}(u(e)x(e)+b(v)+b(w))\] and for extension edges \[-\frac{2}{1+\beta}u(e)x(e)\geq -\frac{2}{1+\beta}(u(e)x(e)+b(v_e)).\]
For a $2+\beta$-approximation, we need
\[\frac{2}{1+\beta}=\beta\Leftrightarrow\beta\in\{-2,1\}.\]
Since $\beta\geq0$, we choose $\beta:=1$ and hence $\alpha:=\frac{2}{3}$ for a 3-approximation.

\section{The integrality gap of the LP}
\label{sec:integrality_gap}
We will now prove that the integrality gap of the LP is $3$. This means that using the approach discussed here, we can not achieve a better approximation guarantee.
\begin{thm}
 The integrality gap of the LP-relaxation given in Section \ref{sec:lp} is at least $3$.
\end{thm}
\begin{proof}
 For an instance $I$ denote by $\opt(I)$ the value of an optimum (integral) solution and by $\optlp(I)$ the value of an optimum LP-solution. We will provide a sequence $I_k$ of instances, such that $\lim\limits_{k\rightarrow\infty}\frac{\opt(I_k)}{\optlp(I_k)}=3$.
 
 Let $0<\epsilon<\tfrac{1}{2}$. For some $k\geq3$, let $G$ be a $k$-star. That is a graph with $k+1$ vertices $\{C\}\cup\{v_1,\ldots,v_k\}$ and edges $\{\{C,v_i\}\mid i=1,\ldots,k\}$. We set $c\equiv0$ and $\gamma:=1$. For all edges $e\in E(G)$, we set $u(e):=\frac{1}{2}$. Finally, we set $b(C):=1-\epsilon$ and $b(v_i):=\epsilon$ for $i=1,\ldots,k$. In order to get to a complete graph, we extend $G$, by adding edges between all pairs $v_i,v_j$ for $i<j$ and set $c(\{v_i,v_j\})=0$ and $u(\{v_i,v_j\}):=1-\epsilon$. Clearly, the resulting $u$ and $c$ are metric. We will denote this instance by $I_{k,\epsilon}$. A depiction of $I_{k,\epsilon}$ is shown in Figure \ref{fig:integrality_example}.
 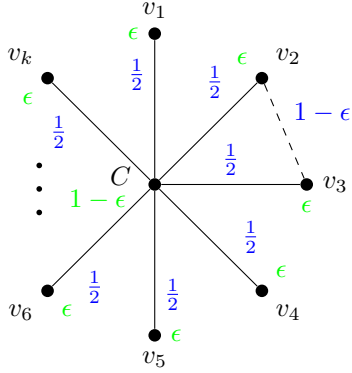
\begin{figure}\centering
 \begin{tikzpicture}
  \node[circle,fill=black,scale=0.5] (C) at (0,0){};
  \node[circle,fill=black,scale=0.5,label=north:$v_1$,label=west:\textcolor{green}{$\epsilon$}] (v1) at (0,2){};
  \node[circle,fill=black,scale=0.5,label=north east:$v_2$,label=north west:\textcolor{green}{$\epsilon$}] (v2) at (1.41,1.41){};
  \node[circle,fill=black,scale=0.5,label=east:$v_3$,label=south:\textcolor{green}{$\epsilon$}] (v3) at (2,0){};
  \node[circle,fill=black,scale=0.5,label=south east:$v_4$,label=north east:\textcolor{green}{$\epsilon$}] (v4) at (1.41,-1.41){};
  \node[circle,fill=black,scale=0.5,label=south:$v_5$,label=east:\textcolor{green}{$\epsilon$}] (v5) at (0,-2){};
  \node[circle,fill=black,scale=0.5,label=south west:$v_6$,label=south east:\textcolor{green}{$\epsilon$}] (v6) at (-1.41,-1.41){};
  \node[circle,fill=black,scale=0.5,label=north west:$v_k$,label=south west:\textcolor{green}{$\epsilon$}] (vk) at (-1.41,1.41){};
  \node[rotate=90,scale=2] at (-1.5,0) {$\cdots$};
  \node at(-0.45,0.12) {$C$};
  \node at(-0.75,-0.2) {\textcolor{green}{$1-\epsilon$}};
  
  \draw (C) -- (v1) node[near end,anchor=east] {\textcolor{blue}{$\frac{1}{2}$}};
  \draw (C) -- (v2) node[near end,anchor=south east] {\textcolor{blue}{$\frac{1}{2}$}};
  \draw (C) -- (v3) node[midway,anchor=south] {\textcolor{blue}{$\frac{1}{2}$}};
  \draw (C) -- (v4) node[near end,anchor=south west] {\textcolor{blue}{$\frac{1}{2}$}};
  \draw (C) -- (v5) node[near end,anchor=west] {\textcolor{blue}{$\frac{1}{2}$}};
  \draw (C) -- (v6) node[near end,anchor=north west] {\textcolor{blue}{$\frac{1}{2}$}};
  \draw (C) -- (vk) node[near end,anchor=north east] {\textcolor{blue}{$\frac{1}{2}$}};
  \draw[dashed] (v2) -- (v3) node[midway,above right] {\textcolor{blue}{$1-\epsilon$}};
 \end{tikzpicture}
 \caption{A picture showing the instance described in the proof of Theorem \ref{sec:integrality_gap}. The solid edges belong to the $k$-star. Edge loads are marked in blue and node loads are marked in green. The dashed edge is an example for the edges added to complete the graph.}
 \label{fig:integrality_example}
 \end{figure}

 In an optimum integral solution to this instance, no edge can be used. This means that $\opt(I_{k,\epsilon})=k+1$. Now we solve the $LP$ using the algorithm from section 2. It will first consider the edges $\{C,v_i\}$ and only afterwards the others. W.l.o.g. we may assume that the edges are considered in the order $e_1:=\{C,v_1\},\ldots,e_k:=\{C,v_k\}$. The first edge will get the value $y(e_1):=2-(1-\epsilon+\epsilon)=1$ and hence \[x(e_1)=\frac{y(e_1)}{1+\frac{1}{2}}=\frac{2}{3}.\] The edge $e_i$ will get the value $y(e_i):=i+1-(i+\epsilon)=1-\epsilon$ and hence \[x(e_i)=\frac{2}{3}-\frac{2\epsilon}{3}      
                                                                                                                                                                                                                                              .\]
 After edge $e_k$ has been considered, the support graph is a tree. This means, that the algorithm will not consider the remaining edges. This shows that \[\optlp(I_{k,\epsilon})=\left|V\right|-\sum\limits_{i=1,\ldots,k}x(e_i)=k+1-\frac{2}{3}-(k-1)\left(\frac{2}{3}-\frac{\epsilon}{3}\right)=\frac{k}{3}+\frac{(k-1)\epsilon}{3}+1.\]
 Setting $I_k:=I_{k,\tfrac{1}{k^2}}$, we get
 \[\lim\limits_{k\rightarrow\infty}\frac{\opt(I_k)}{\optlp(I_k)}=\lim\limits_{k\rightarrow\infty}\frac{k+1}{\frac{k}{3}+\frac{(k-1)}{3k^2}+1}=3\]
\end{proof}
\begin{rem}
 Adding inequalities that forbid edges that can not be taken in an integral solution will not help to lower the integrality gap. We can replace the central vertex $C$ in the proof of Theorem 9 by a $K_l$ for some $l\in\mathbb{N}$ with small edge and vertex loads that add up to $1$. This will only change the ``$+1$'' by a constant depending on $l$ and hence not change the asymptotic ratio of $\opt$ to $\optlp$.
\end{rem}
Together with the upper bound of $3$ given by the analysis of the algorithm, we can conclude:
\begin{cor}
The integrality gap of the LP is $3$.
\end{cor}

\printbibliography

\end{document}